\documentclass[conference]{IEEEtran}
\makeatletter
\def\ps@headings{%
\def\@oddhead{\mbox{}\scriptsize\rightmark \hfil \thepage}%
\def\@evenhead{\scriptsize\thepage \hfil \leftmark\mbox{}}%
\def\@oddfoot{}%
\def\@evenfoot{}}
\makeatother
\usepackage{times}
\usepackage{amsmath}
\usepackage{amssymb}
\usepackage{url}
\usepackage{wrapfig}

\usepackage{color, array, colortbl}
\usepackage{comment}

\usepackage{graphicx}
\usepackage{placeins}
\newtheorem{lemma}{Lemma}

\newtheorem{proof}{Proof}
\newtheorem{definition}{Definition}

\newcommand{\DCost}{\text{\sc Cost}}

\newcommand{\dist}{\text{d}}
\newcommand{\cancel}[1]{}
\def\abs#1{\lvert #1 \rvert}
\def\BC#1#2{BCube(#1,#2)}

\newcommand{\BestNeighborI}{{\sc BestNeighbor-Indirect} }
\newcommand{\BestNeighborD}{{\sc BestNeighbor-Direct} }

\newcommand{\LCost}{\text{\sc Lcost}}

\begin{document}

\title{Simple Destination-Swap Strategies for\\Adaptive Intra- and Inter-Tenant VM Migration}


\author {Chen Avin$^1$, Omer Dunay$^1$, Stefan Schmid$^2$\\
{\small $^1$ Ben Gurion University, Beersheva, Israel;~~ $^2$ TU Berlin \& T-Labs, Berlin, Germany}\\
{\small \texttt{\{avin,dunay\}@cse.bgu.ac.il}; \texttt{stefan@net.t-labs.tu-berlin.de}}\\}

\date{}

\maketitle \thispagestyle{empty}
\sloppy

\begin{abstract}
%
This paper investigates the opportunities and limitations of
adaptive virtual machine (VM) migration to reduce communication costs in a virtualized
environment. We introduce a new formal model for the problem of online VM migration in two scenarios: (1) VMs can be migrated arbitrarily in the substrate network;
e.g., a private cloud provider may have an incentive to reduce the overall communication cost in the network.
(2) VMs can only be migrated within a given tenant; e.g., a user that was assigned a set of physical machines
may exchange the functionality of the VMs on these machines.

We propose a simple class of \emph{Destination-Swap} algorithms which are based on an aggressive collocation strategy
(inspired by splay datastructures) and which maintain a minimal and local amount of per-node (amortized cost) information
to decide where to migrate a VM and how; thus, the algorithms react quickly to changes in the load.
The algorithms
come in two main flavors, an indirect and distributed variant which keeps existing VM placements local, and a direct variant
 which keeps the number of affected VMs small.

We show that naturally, inter-tenant optimizations yield a larger
potential for optimization, but generally also a tenant itself can improve its embedding. Moreover, there exists an interesting
tradeoff between direct and indirect strategies: indirect variants are preferable under skewed and sparse communication patterns
due to their locality properties.
\end{abstract}

%
%

\renewcommand{\baselinestretch}{.91}

\section{Introduction}

Virtualization is perhaps the main innovation motor in today's networks. In particular, most datacenter resources have
become fully virtualized, and resource slices are carved out dynamically for the different applications.
The basic unit in such datacenters are \emph{virtual machines (VMs)}: Ideally, a VM appears to be a dedicated physical machine,
but in reality the VM may share the underlying physical machines with other VMs. Moreover, since the virtualization layer decouples
the virtual services and resources from the constraints of the underlying physical infrastructure, a VM can be \emph{migrated}
seamlessly to other physical machines.

Migration introduces new flexibilities on how to manage a given resource network. While the allocation (and isolation)
of CPU and memory resources is fairly well understood today, the \emph{access} to these resources
has often been treated as a second class citizen. However, it has recently been argued that connecting VMs with explicit bandwidth and hence communication guarantees,
 can reduce the variance and duration (price) of an execution.~\cite{mesos,talk-about,qclouds,amazon-per,EC2over}
Moving frequently communicating VMs closer together can save network bandwidth (or even energy~\cite{Shang2010Energy-aware}) and improve the predictability of the execution. Automatic collocation may also enhance the privacy of the communication~\cite{getoff}, especially in wide-area networks.

In this paper, we attend to a generic setting consisting of a physical network which is shared by multiple tenants.
The physical network connects different physical machines. It may represent a datacenter, but may also be
a wide-area network spanning globally distributed resources or ``micro-datacenters'' (e.g., the POPs or even street cabinets
of an Internet Service Provider, see also the trend towards network functions virtualization~\cite{nfv}). In the following, we will often refer to this network as the \emph{substrate} or \emph{host graph}.

Each tenant has access to a set of VMs. These VMs can be mapped arbitrarily on the host graph (maybe subject to some specification constraints
in case of a wide-area network). In order to complete their tasks, the VMs of a tenant need to communicate. Accordingly, in this paper,
we will think of the VMs and their interactions as a graph as well: each tenant describes a
(dynamic) \emph{guest graph} which is embedded on the host graph.

This paper studies migration algorithms which ensure that the frequently communicating VMs of a guest graph are migrated
together adaptively. We distinguish between two migration scenarios:
\begin{enumerate}
\item
 \emph{Inter-tenant VM migration:} In this scenario, VMs can be migrated arbitrarily on the host graph. For example,
a private datacenter provider may have an incentive to globally optimize the allocation of all VMs across all tenants, in order
to reduce network loads and improve performance. Or in the context of wide-area networks (WANs) and ISPs
supporting network functions virtualization, a provider may optimize the applications
of all its customers by migrating critical VMs geographically closer and hence reduce latencies.

\item \emph{Intra-tenant VM migration:} In this scenarios, the set of physical machines assigned to a tenant is fixed.
However, the tenant can re-assign the VMs (e.g.,~\cite{intra-tenant}) and functionality of its own application among these physical machines. For instance,
while a monopolistic public cloud provider may not have an incentive (or is not allowed) to optimize the VM mapping,
a user may improve its application performance by collocating frequently communicating machines.
\end{enumerate}

\textbf{Contribution.} This paper makes the following contributions.
\noindent (1) We initiate the study of adaptive inter- and intra-tenant migration strategies and present
a simple formal model which allows us to reason about and analytically evaluate different online algorithms.
In particular, we pursue a conservative and online approach and assume that we do not have any
a priori knowledge of the guest graph traffic matrices and their evolution over time. Thus, our approach also captures
settings where the tenant specifies rough communication requirements (such as \emph{Virtual Clusters} or \emph{Virtual-Oversubscribed-Clusters}~\cite{oktopus}).

\noindent (2) We introduce a simple class of \emph{Destination-Swap} migration algorithms which only require
a very small amount of per-node state information, namely a \emph{local amortized cost}. This simple approach allows us to focus on the main properties and
challenges of adaptive migration. The \emph{Destination-Swap} algorithms are based on an aggressive collocation
approach, i.e., communicating VMs are migrated together. This is in the spirit of classic self-adjusting
distributed datastructures such as splay trees. The algorithms come in two main flavors: An \emph{indirect} swapping
strategy where VMs are migrated iteratively to each other; this preserves (communication) locality and can be seen as a distributed
computing approach. A \emph{direct} swapping strategy where two VMs are collocated directly, i.e., without moving other VMs
along the path; although a single VM may be globally
displaced, the direct strategy has the advantage that only one external VM is affected (local impact).
Generally, while these algorithms can be used for any host graph topology,
we in this paper will mostly focus on the BCube datacenter topologies.

\noindent (3) We study different variants of \emph{Destination-Swap} migration algorithms under different communication patterns, and
show that they have some interesting properties. In particular, we find that
smart migration algorithms can indeed reduce the communication cost, especially in the inter-tenant
scenario; in the more constrained intra-tenant scenario, the amortized (communication) costs
can be lowered too, but only to a lesser extent. Moreover, we find that the algorithms adapt
relatively quickly to new communication patterns and that there are interesting
differences between direct and indirect swapping approaches. In particular, indirect variants are preferable
if communication patterns are sparse and guest graphs tree like, and if the communication
frequencies between VMs is subject to a higher variance; otherwise, direct variants perform better
due to the limited impact on other VM locations.

We understand our paper as a first step to shed light on the online VM migration problem,
and we kept our algorithms as simple as possible. In particular, we believe that our work
opens several interesting directions for future research, and more sophisticated approaches.


\section{Online VM Rearrangement}\label{sec:model}

We propose the following simple formal model to reason about dynamically reconfigurable virtual networks.
We consider a physical network $H=(M \cup S,L)$, the \emph{host graph} (e.g., a datacenter), where $M=\{m_1, m_2, \dots\}$ represents the physical machines (or simply \emph{hosts}),
$S= \{s_1, s_2, \dots\}$ represents the switches (or routers) connecting the hosts, and $L$ represents the physical links. A link (or edge) may connect a switch with a host and/or two hosts directly. We will assume that communication over a physical link comes at a certain cost (network load or latency), but we do not assume any strict link capacity constraints.

The physical network hosts a set of applications $\mathcal{A}=\{A_1,\ldots,A_k\}$, from $k$ different \emph{tenants}. (In the following, we will
simply assume that each application corresponds to one tenant, and will treat the terms as synonyms.) Each application $A_i$
consists of a set of virtual machines, i.e., $A_i=\{v^{(i1)},\ldots,v^{(ij)}, \dots \}$. We will refer to the cardinality of a set $X$ by $\abs{X}$.
For simplicity and ease of presentation,
in the following, we will focus on a scenario where each machine $m \in M$ can host exactly
one VM.

The \emph{arrangement function} $\lambda$ describes the mapping of the VMs to physical machines. That is, $\lambda(v)$ denotes the physical machine to which a given VM $v$ is mapped. Analogously, the VM hosted by a machine $m$ is denoted by $\lambda^{-1}(m)$; $\lambda^{-1}(m)=\bot$ means that $m$ does not host any VM.

The basic objective of the migration algorithms studied in this paper is to re-arrange VMs adaptively, in order to reduce the amount of communication. Since the communication pattern may change over time, mapping decisions may be reconsidered repeatedly. Therefore, we will use a time index $t$, and let $\lambda_{t}$ refer to the arrangement function at time $t$: $\lambda_t(v)$ is the mapping of $v$ at time $t$.

Let us use $\dist(u,v)$ to denote the distance between $\lambda(u)$ and $\lambda(v)$ in $H$: the number of physical links $\ell \in L$ needed to connect $\lambda(u)$ and $\lambda(v)$ along the \emph{shortest path} in the host graph. Since the arrangement function changes over time, we also define the temporal version $\dist_{t}(u,v)$: the shortest distance between $\lambda_{t}(u)$ and $\lambda_{t}(v)$.

We assume that the different VMs of tenant/application $A_{i}$ need to communicate in order to fulfill their task. Although our algorithms
are applicable more generally, in the following, in order to study the reactivity and convergence properties of the migrations,
we will often make the assumption that the interactions between the VMs follow a \emph{statistical distribution}:
this distribution is not known in advance, and interactions are sampled from this distribution independently and at random over time.
Accordingly, we can represent each application $A_{i}$ as a \emph{(tenant) guest graph} $G_{i}=(A_{i},E_{i},w_{i})$,
where $A_{i}$ is the set of virtual machines and
$E_i$ represents the interactions between the VMs. The weight $w(e)$ of an edge $e=(u,v) \in E_i$ describes the
frequency of the interactions between the two VMs $u,v \in A_{i}$ according to this distribution.
The total set of \emph{all} applications $\mathcal{A}$, the union of all tenant guest graphs, is represented as an \emph{(overall) guest graph} $G(V,E, w) = (G_{1} \cup G_{2} \cup \ldots \cup G_{k})$
where $V=\bigcup_i A_i$, $E=\bigcup_i E_i$ and $w=\bigcup_i w_i$.
In other words, the overall guest graph $G$ is a combination of \emph{independent} connected components.

For our evaluation, we will often assume that the guest graph $G$ is fixed. However, as we will see,
our migration algorithms have a low converge time and hence quickly adapt to a new structure of $G$.
We assume a conservative perspective and assume that the network controller (the orchestration manager) does not know anything
about the guest graph $G$ of the application in advance.
Rather, the communication pattern between VMs is revealed over time, and we are hence interested in \emph{online} migration algorithms $\textsc{Alg}$.

Generally, the input to the online algorithm $\textsc{Alg}$ is a sequence $\sigma=(\sigma_1, \sigma_2, \dots)$ of communication requests $\sigma_t = (u_t,v_t)$ (coming, e.g., from a fixed guest graph $G$ and $u_t,v_t \in E$, but it may also be arbitrary). After each such request, $\textsc{Alg}$ is allowed to migrate different VMs, i.e., redefine the mapping $\lambda_t$.
Of course, this request sequence approach is simplistic and just serves for the modeling: we do not imply that migration algorithms should happen on a per-request basis. Rather, in practice, a request may be defined by a certain communication volume (e.g., 5GB), possibly over a certain time interval.

Let the migration cost (i.e., the number of migrations) of $\textsc{Alg}$ for a request $\sigma_t$ at time $t$ be denoted as $\rho_{t}$. 
The main yardstick for our evaluation is the amortized cost of communication and VM migration for a given host graph, algorithm and request sequence:

\begin{definition}[Amortized (Communication) Cost]\label{def:AmortizeCost}
Given a host graph $H$, a migration algorithm $\textsc{Alg}$ and a sequence of communication requests $\sigma$ (e.g., from a guest graph $G=(A,E,w)$), the \emph{amortized cost} is defined as:
\begin{align}
\DCost (H, \textsc{Alg}, \sigma) = \frac{1}{|\sigma|} \sum\limits_{t=1, (u,v)\in \sigma_t}^{|\sigma|} (\dist_{t}(u,v) + \rho_{t})
\end{align}
\end{definition}
Our goal is to find an \textsc{Alg} that minimizes \DCost.

We will simplify the model by normalizing the migration cost to the cost of a one hop routing request, i.e., we consider the amortized cost defined as
$
\DCost (H, \textsc{Alg}, \sigma) $ $= 1/|\sigma| \sum_{t=1, (u,v)\in \sigma_t}^{|\sigma|}$ $\dist_{t}(u,v)
$
and we attend to the problem of finding an algorithm that minimizes the above amortized cost.

	\begin{figure}[h]
				\centering
				\includegraphics[width=.45\linewidth]{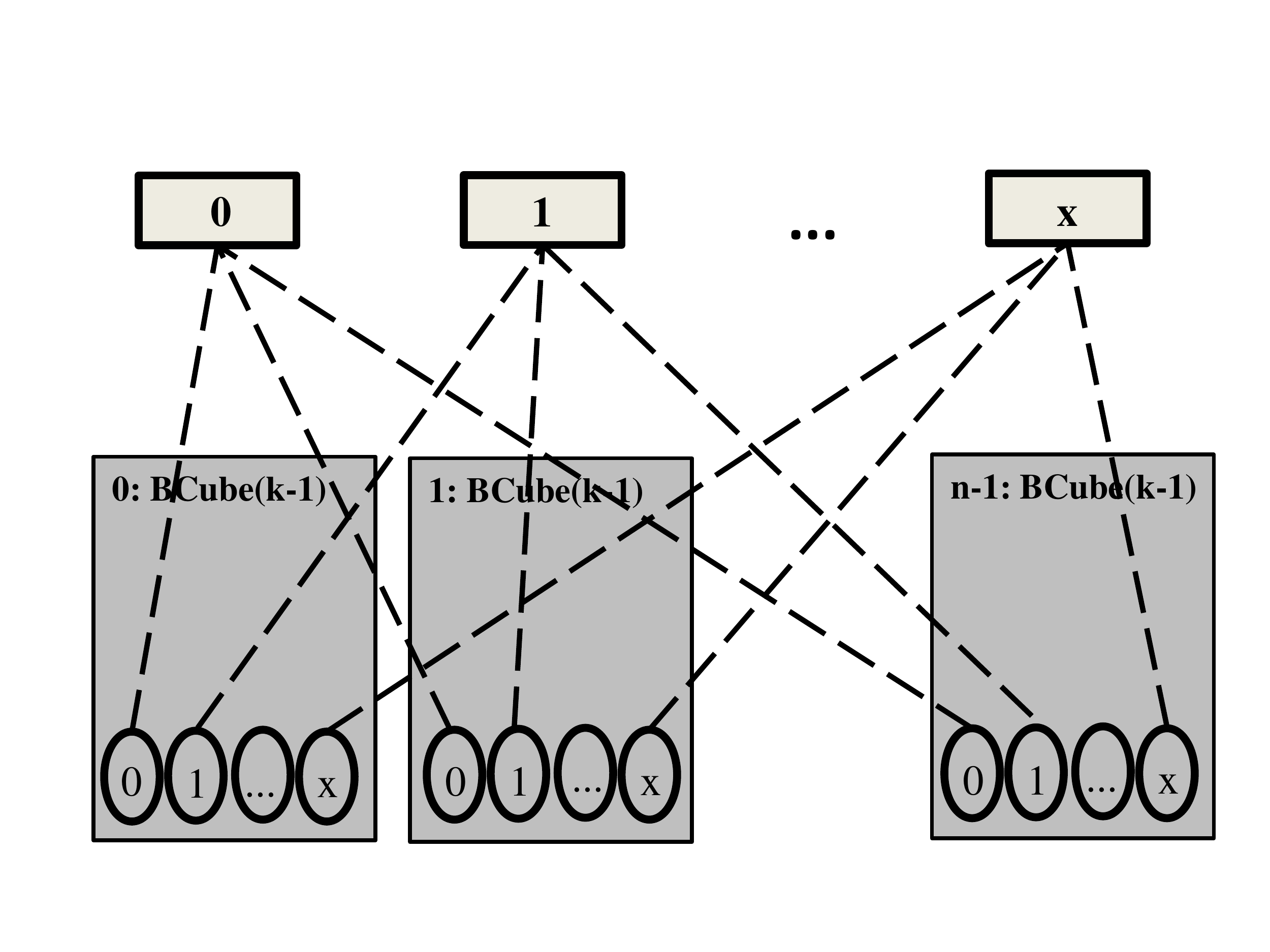}~~~~\includegraphics[width=.45\linewidth]{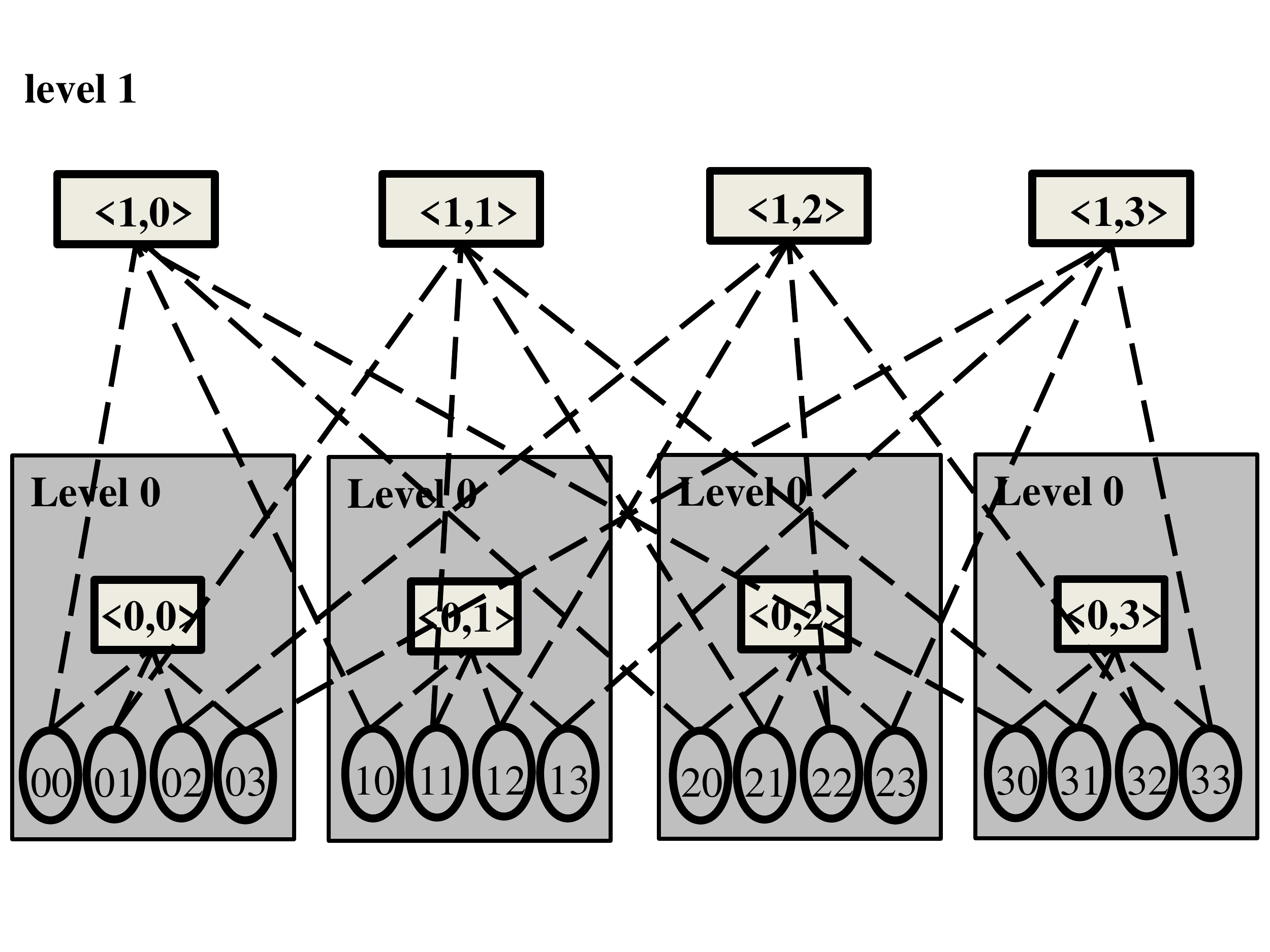}\\
				\caption{(a) BCube is defined recursively: A $\BC{n}{k}$ is constructed from $n$ $\BC{n}{k-1}$) and $n^k$ $n$-port switches. (b) The $\BC{4}{1}$.}
				\label{fig:bcube}
				\end{figure}

\subsection{The BCube Architecture}
\label{sec:Bcube}

Although our algorithms are applicable in general graphs, this paper focuses on the special family of \emph{BCube}~\cite{bcube2009} host graph topologies.
The \emph{BCube} is a new network architecture exhibiting a hypercubic topology. It is tailored for modular datacenters.

There are two types of devices in a \emph{BCube}: servers (hosts) with multiple ports, and switches that connect a constant number of servers.
The BCube topology $\BC{n}{k}$ has $k+1$ levels and uses $n$-port switches; hosts have $k+1$ ports.
The structure is defined recursively:
A  $\BC{n}{0}$ describes a ``star network'' where $n$ hosts connect to an $n$-port switch.
A $\BC{n}{1}$ is constructed from $n$ different $\BC{n}{0}$ connected by $n$ different $n$-port
switches. More generically, a $\BC{n}{k}$ ($k \ge1$) is constructed from $n$ $\BC{n}{k-1}$ and $n^k$ $n$-port switches.

Each server in a  $\BC{n}{k}$ has $k + 1$ ports, which are numbered from level-0 to level-$k$. It is easy to see that a BCube has
$N = n^{k+1}-1$ servers and $k+1$ levels of switches, with each level
having $n^k$ $n$-port switches. Figure~\ref{fig:bcube} presents an example of a BCube.


The BCube guarantees that switches only connect to servers and never directly connect to other switches.
As a direct consequence, for our purposes, we can treat the switches as dummy crossbars that connect several neighboring servers and let servers relay traffic for each other. In addition, the BCube has several nice properties that makes it attractive for use as a datacenter: First, its topology is very robust and has
$k+1$ disjoint paths between any two hosts. If we ignore the switches, the BCube topology is essentially a generalized
hypercube~\cite{generalizedHyper}. In order to address host and perform greedy routing, a  simple vector of size $k+1$ is used. The distance between any two hosts is
given by their \emph{Hamming distance}, and hence the network diameter is $k+1$ ~\cite{bcube2009}.

\section{The Destination-Swap Algorithms}\label{sec:algoanal}

This section introduces the class of \emph{Destination-Swap} migration algorithms that adaptively improve
the VM embedding over time. The appeal of these algorithms comes from their simplicity: (1) These algorithms are based on very little
state information and do not perform any long-term statistical analysis or a complicated pattern learning; this allows the algorithms to stay reactive and
adapt to new patterns quickly. (2) The simplicity of these algorithms facilitates the formal study of the benefits and inherent challenges of
VM migrations.

We will first present the general concepts of Destination-Swap and then consider some specific algorithms in more detail.

\cancel{
The first algorithm we denote as {\sc MeetMiddle} is very simple, Upon a communication request $(u,v)$ perform local swaps (migrations) between $u$ and one of it's neighbors on the shortest path to $v$ (select that neighbor randomly with uniform distribution). make the same operation for $v$. These swaps are performed until $u$ and $v$ are neighbors.}

\subsection{Algorithmic Framework}

As the name suggests, the \emph{Destination-Swap} migration algorithms
consist of two modules: one to select a ``destination'' where a given VM should be migrated to,
and one to decide on the migration or ``swapping'' strategy to reach the selected destination.

Concretely, upon a request $(u,v)$ at time $t$,  in order to reduce the communication cost of future requests,
the algorithm decides to migrate $u$ and $v$ closer to each other. To do so, the algorithm may for example select a server $\mu$ as a \emph{destination} host for $v$ and leave $u$ fixed (\emph{destination strategy}). In order to migrate $v$ to $\mu$, the algorithm may perform several rearrangements (\emph{swapping strategy}), involving also other VMs. Eventually, $\lambda(v) = \mu$ and $\dist_{t+1}(u,v) < \dist_t(u,v)$.

We will describe our algorithms from an inter-tenant perspective. However, by restricting migrations to the physical machines of the given tenant only (and otherwise not migrate at all), the algorithms can be adapted for the intra-tenant scenario. Moreover, although migration decisions can be performed globally, our algorithms
could also be seen from a distributed computing perspective: as we will see, our indirect swapping algorithm only involved local interactions, and our direct
swapping algorithm only involves three VMs.



\subsection{Destination Strategy}

In general, we seek to move $u$ and $v$ closer together upon each communication request
$(u,v)$, either by moving one of the two nodes or both.
Concretely, all the algorithms presented in this paper will make $u$ and $v$ immediate
neighbors after the request. The intuition for this rather aggressive
strategy comes from the related approaches used for dynamic splay trees~\cite{Sleator:1985:SBS:3828.3835} and their
distributed generalizations~\cite{ipdps13}.
\begin{wrapfigure}{l}{0.6\columnwidth}
	\centering
\centering
				\includegraphics[width=.55\columnwidth]{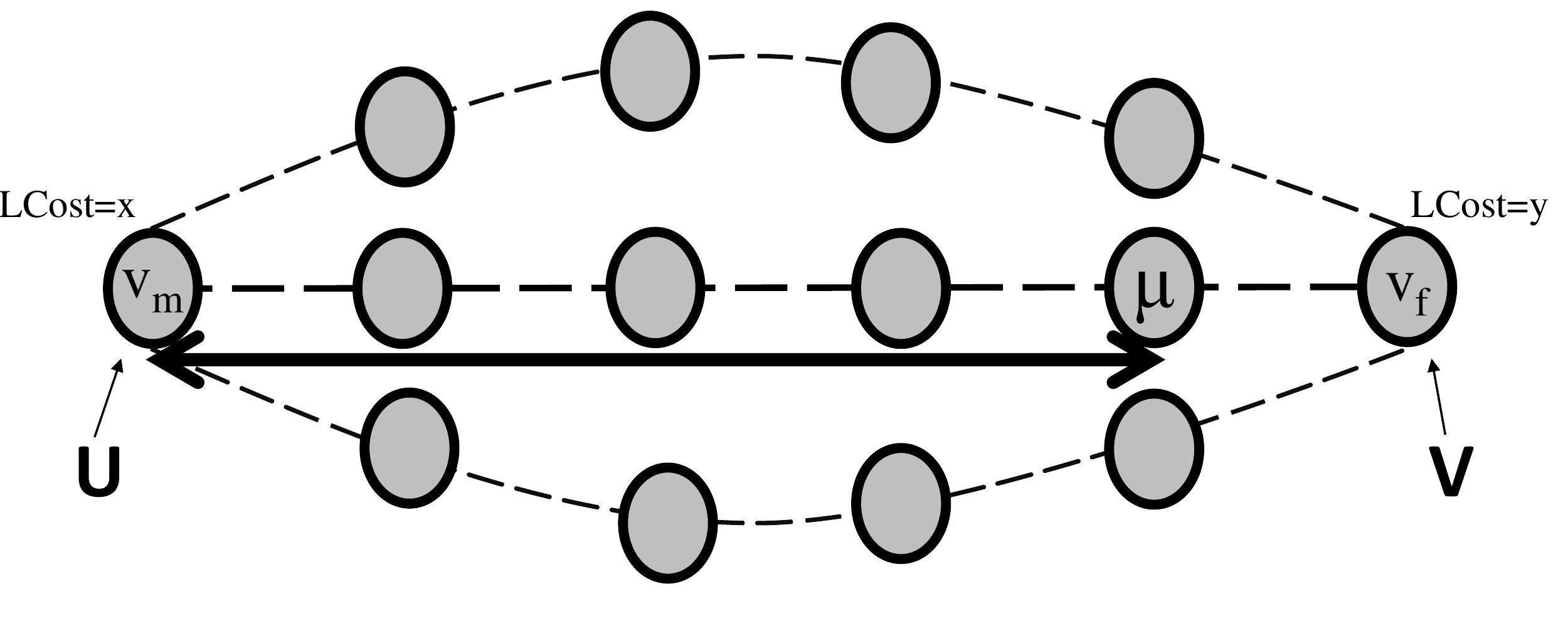}\\
				\caption{Principle of Destination-Swap Algorithms}
				\label{fig:dswap}	
\end{wrapfigure}
We distinguish between four basic methods: {\sc MeetMiddle}, {\sc Random}, {\sc BestSwitch} and {\sc BestNeighbor}.
While the first method does not distinguish between the two interacting VMs $u$ and $v$, the other three methods are based on the amortized cost states.
In the following, we elaborate more on these methods.

Methods {\sc Random}, {\sc BestSwitch} and {\sc BestNeighbor} only choose \emph{one} of the two VMs $u$ or $v$ for migration:
Let us denote the to be migrated VM by $v_m$  and the fix VM by $v_f$.  The migration destination, henceforth denoted by $\mu$, is
always selected to be a neighbor of $\lambda(v_f)$; accordingly, after the migration, $u$ and $v$ are direct neighbors in the host graph.

The decision of which of the two VMs $u$ and $v$ should be migrated, is based on a \emph{local amortized cost} criterion. The basic idea of this
criterion is that a VM that is already \emph{close} to all its communication partners, does not need to migrate. 
For such a VM the local amortized cost should be low. Contrarily, the VM which is located at a suboptimal position in the host graph and which is still far from their partners, the cost should be high and migrating it to a different position may be beneficial. Formally,
let $\sigma_{v}\subseteq \sigma$ denote the set of communication requests which include the VM $v$ (either as a communication source or destination).
\begin{definition}[Local Amortized Cost]\label{def:LocalAmortizeCost}
The \emph{local amortized cost} of a VM $v$ is defined as
\begin{align}
\LCost(v,\sigma) = \frac{1}{\abs{\sigma_{v}}} \frac{1}{\log (\abs{\sigma_{v}})} \sum\limits_{t\in \sigma_{v}} \dist_{t}(u,v)
\end{align}
\end{definition}

Note that in contrast to Definition~\ref{def:AmortizeCost} which defines the total amortized cost of the \emph{entire} network,
 Definition~\ref{def:LocalAmortizeCost} can be calculated locally by each VM and it defines a per-node cost criterion to decide which of the two communication end-points $u$ or $v$ is already located at a strategically better position and should hence not be migrated.
We use an additional logarithmic factor $1/\log (|\sigma_{v}|)$ to give more weight to a frequently communicating VM.


Algorithms {\sc Random}, {\sc BestSwitch} and {\sc BestNeighbor} migrate the VM with the higher local amortized cost while the one with lower cost stays at the same location.



Henceforth, for a communication request $(u,v)$, let the ``migrating node''  be $v_m = \arg\min_{u,v} \{ \LCost(u),\LCost(v)\}$;
analogously, let the ``fixed node'' be $v_f= \arg\min_{u,v} \{ \LCost(u),\LCost(v)\}$. If $\LCost(u)=\LCost(v)$, the tie can be broken arbitrarily.

After the choice of $v_f$ and $v_m$ is made, \emph{Destination-Swap} algorithms need to decide
on the destination: we select the the destination host $\mu$ as a neighbor of $\lambda(v_f)$. Figure~\ref{fig:dswap} illustrates the situation.

In order to describe the choice of $\mu$, we define another cost function called {\sc Sc}. It measures the amount of communication among a \emph{set} of hosts.
Informally, for a set of hosts $Q$, the score {\sc Sc}$(Q)$ counts all the communication requests between the VMs that are currently hosted by $Q$. Formally,
let $\sigma_{u,v}$ be the communication requests for which the communication partners VMs are $u$ and $v$.
\begin{definition}[{\sc Sc}]\label{def:scoreServers}
For a set of hosts $Q$, let {\sc Sc}$(Q)$
\begin {align}
\text{\sc Sc}(Q) =  \sum_{m, m' \in Q} \abs{\sigma_{(\lambda^{-1}(m), \lambda^{-1}(m'))}}
\end {align}
\end{definition}

For now let $m= \lambda^{-1}(v_f)$ be the host server of $v_f$.  The first three algorithm we consider are:
\begin{enumerate}
\item {\sc Random:} We select $\mu$ as a random neighbor of $m$.
\item {\sc BestSwitch:} We select $\mu$ from the ``best switch'' that $m$ is connected to according the following rule.
Recall that $m$ is connected to $k+1$
switches henceforth denoted by $\{T_0,\ldots,T_k\}$, each of which belongs to a different level of the BCube. For each switch $T_i$ we compute its
current {\sc Sc}$(Q_i)$, where $Q_i$ is the set of servers attached to switch $T_i$, and its score if $v_m$ would be connected to it (replacing the least communicating VM).
We select, as the best switch, the switch with the largest increase in score, and the corresponding machine $\mu$ of that switch.
Note that $\mu$  must be a neighbor, since the nodes belong to the same switch.
\item {\sc BestNeighbor:} Let $\mathcal{N}(m)$ denote the set of neighbors of $m$ in the host graph $H$. We select $\mu$ as the neighbor of $v_f$ for which
migrating $v_m$ to  $\mu$ increases {\sc Sc}$(\mathcal{N}(\mu))$ the most.
\end{enumerate}

The forth algorithm is inspired by similar strategies for splay tree networks (see, e.g.~\cite{ipdps13}) and serves as the
baseline performance.
\begin{enumerate}
\item[4)]{\sc MeetMiddle}
 We migrate both $u$ and $v$ such that they become neighbors \emph{in the middle} on an arbitrary shortest path between them.
 The two communicating VMs are treated as ``equal''.
\end{enumerate}

\subsection{Swapping Methods}

After having selected the destination node $\mu$ for $v_m$, we need to decide \emph{how to move} $v_m$ to $\mu$.
We distinguish between two main strategies.
\begin{enumerate}
	\item \emph{Direct}: Migrate $v$ to $\mu$ directly, i.e., exchange the locations of VMs $v$ and $\lambda^{-1}(\mu)$.

	\item \emph {Indirect}: Swap $v$ iteratively with neighbors along an arbitrary shortest path to $m$, until $\lambda(v)$ is at distance
	two from $m$. Then swap directly the VMs $v$ and $\lambda^{-1}(\mu)$.
\end{enumerate}

\begin{wrapfigure}{l}{0.5\columnwidth}
	\centering
\centering
				\centering
				\includegraphics[width=.48\columnwidth]{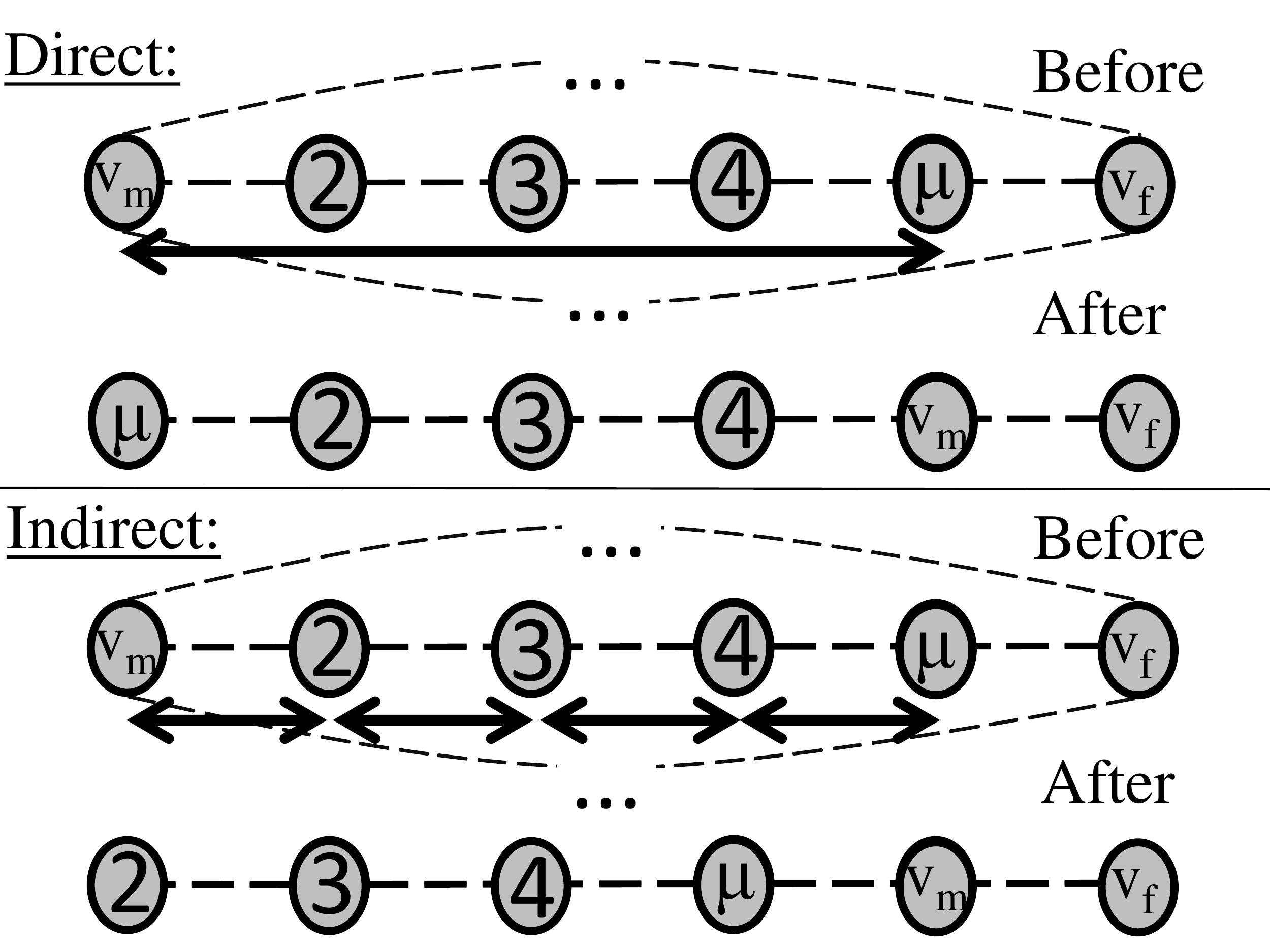}\\
				\caption{Direct vs Indirect Swapping}
				\label{fig:indirect}
\end{wrapfigure}
Intuitively, the indirect approach seeks to keep the embedding local, by not migrating $\lambda^{-1}(m)$ too far; in this light,
it can be seen as a distributed algorithm. The direct approach
on the other hand has the advantage that while $\lambda^{-1}(\mu)$ is globally displaced, less VMs are affected by the change:
while the indirect approach migrates \emph{all} VMs along the shortest path, the direct approach changes the locations of two VMs only.
Figure~\ref{fig:indirect} illustrates the situation.


\subsection{Summary}

A Destination-Swap algorithm is a combination of a destination and a swapping strategy. For example, \BestNeighborD is an algorithm that uses \emph{BestNeighbor} as its destination method and \emph{direct} as its swapping method.

\section{Evaluation}\label{sec:simulations}

We analyze the behavior and performance of the \emph{Destination-Swap} migration algorithms formally and
by simulation.

\subsection{Guest Graph Model}

In order to model the communication patterns, in our evaluation we focus on the following \emph{guest graphs}.
We assume that the overall communication pattern or \emph{(overall) guest graph} consists of multiple connected components, the guest graphs of the tenants
(the \emph{(tenant) guest graphs}); in the following, when it is clear from the context to which type of guest graph we refer, we will omit the \emph{overall} and \emph{tenant}
specifier.

If not stated differently, we will assume that the total number of VMs equals the number of physical servers, that is,
the sum of all connected components perfectly covers all servers. (Recall that we assume that a physical server has sufficient capacity to host exactly one VM, and that
we do not assume any strict capacity caps on the links.)

For simplicity, we will focus on scenarios where all connected components have the same topology.
For example, for a host graph $\BC{3}{7}$, we have $6,561$ servers; if guest graphs (the connected components)
are of size $729$ VMs, we will embed nine guest graphs of a given topology ($9\cdot 729 = 6,561$).

We concentrate on the following archetypical connected guest graph topologies: (1) complete graphs (a.k.a.~\emph{cliques}) describing
an \emph{all-to-all communication} pattern; (2)
star graphs describing a \emph{one-to-all communication} pattern; (3) for a host graph \BC{$n$}{$k$}
we consider smaller BCube guest graphs \BC{$n'$}{$k'$}:
we choose $k'$ for all possible values $\{0,\ldots,k-1\}$ and $n' = n$; (4) a guest graph consisting of a set of VM pairs connected by single edge.

The traffic matrices of (1) and (2) are standard and studied frequently~\cite{podc11}; the motivation for (3) is that it constitutes an interesting scenario
between the two extremes (1) and (2), and because, due to the recursive structure of the BCube, small BCubes (``sub-cubes'') can always be \emph{perfectly} embedded
in larger BCubes, i.e., there exists a mapping with amortized costs 1. The scenario can hence also be used as a simple baseline / cost lower bound to evaluate an embedding or migration algorithm. Finally, option (4) defines a set of VM pairs which need to be matched up, independently of other VMs. Since matching graphs can always be embedded perfectly on any host graph due the independence, it constitutes a natural test case for the migration algorithms.

In the following, let $G^{(K_{x})}$ denote a (overall) guest graphs in which all connected components are complete graphs $K_{x}$ consisting of $x$ nodes. Similarly, let $G^{(S_{x})}$ denote the guest graphs in which all connected components are star graphs $S_{x}$ with $x$ nodes, and let $G^{(\BC{n'}{k'})}$ be a guest graph in which all connected components are \BC{$n'$}{$k'$} graphs. Finally, let us refer to the matching graph by $G^{(M)}$.

A guest graph can come in two flavors: \emph{weighted} and \emph{unweighted}. In the \emph{unweighted} variant, all interactions have the same frequency.
For the weighted variant, we will consider randomized weight distributions, e.g., where the weight is chosen uniformly at random from the range $[1,\ldots,N]$, \emph{for each node} independently; the frequency of an edge is then simply the product of the frequencies of its incident VMs (the so-called \emph{product distribution}).

\subsection{Formal Analysis}

This section provides some analytical insights into what can and cannot be achieved by \emph{Destination-Swap} algorithms.
We first derive a formula for the embedding cost in an unoptimized setting, namely where VMs are mapped randomly to the
BCube servers. The setting serves as a simple and worst-case reference point, from which our migration algorithms
seek to improve the VM embeddings. Subsequently, we show that our algorithms have the desirable property, that
they can only improve the embedding, and never increase the amortized costs under a matching communication pattern. Finally, we derive a cost lower bound
on the optimal possible allocation by any migration algorithm for the corresponding guest graph.

\subsubsection{Baseline Performance}

As an initial placement and as a simple baseline for the migration algorithms, we
will sometimes consider a setting where VMs are mapped randomly to the physical machines.
In such a random initial setting, the expected communication cost of a given node pair (under any guest graph!) can
be computed as follows in a BCube.
\begin{lemma}\label{lemma:No-Algo}
Given a host graph $H=\BC{n}{k}$ and an arbitrary guest graph $G$, consider a mapping function $\lambda$ which assigns each VM to a machine of the BCube chosen uniformly at random.
Then, the expected communication cost for any VM pair \BC{$n$}{$k$} is given by
$(k+1)(n-1)/n$.
\end{lemma}
\begin{proof}
Let $X(u,v)$ be a random variable representing the distance of the shortest path between $u$ and $v$.
Recall that in a BCube, the distance between two nodes is given by their Hamming distance.
Fix a specific server $m$ in \BC{$n$}{$k$}. Since the length of the identifiers of the \BC{$n$}{$k$} nodes is $k+1$,
there are $k+1\choose i$
addresses with a Hamming distance of $i$ from $m$.
For every index, $(n-1)$ values can be chosen, and hence the number of servers at distance $i$ from node $m$ is
$(n-1)^i {k+1 \choose i}$
Due to the symmetries of \BC{$n$}{$k$}, this is also the formula for the number of servers at distance $i$ from any given server.
Let us refer to the diameter of the graph by $D$ and recall that
 the diameter of \BC{$n$}{$k$} is $k+1$. Then the expected shortest path from any given $v$ is:

 \begin{footnotesize}
\begin {align*}
 E[X(u,v)] &= \sum\limits_{i=1}^{D} i\cdot P[\dist(u,v) = 1]
= \frac{\sum\limits_{i=1}^{k+1} (i\cdot\dist_i(u,v))}{n^{k+1}} \\
&= \frac{\sum_{i=1}^{k+1} i{k+1\choose i}(n-1)^i}{n^{k+1}}
= \frac{(k+1)(n-1)}{n}
\end {align*}
 \end{footnotesize}
\end{proof}

\subsubsection{Monotonic Improvement Property}

Despite their simplicity, our algorithms feature some basic guarantees.
For example, our algorithms pass the \emph{monotonic improvement test} under
matching communication patterns: if VMs communicate in a pair-wise fashion,
i.e., subject to $G^{(M)}$, the amortized communication cost can only become lower over time.
The proof is simple and omitted due to space constraints.
\begin{lemma}\label{lemma:Matching}
Given a host graph $H = \BC{n}{k}$ and a guest graph $G = G^{(M)}$,
\emph{Destination-Swap}
algorithms can only reduce the amortized communication cost over time.
\end{lemma}
%
%

\subsubsection{Lower Bound}

Of course, there are inherent limitations on what can be achieved by embedding optimizing algorithms.
For clique and star like guest graphs, bounds can be computed from cuts and Huffman coding arguments, see e.g.~\cite{ipdps13}.
However, note that optimal embeddings are possible in our sub-cube guest graphs $G^{B{n'}{k'}}$. We will indicate
this lower bound in the figures of our simulation.

\begin{lemma}
 The guest graph $G=G^{(\BC{n'}{k'})}$ can be perfectly embedded in the host graph $H=\BC{n}{k}$: the amortized cost (Definition~\ref{def:AmortizeCost}) is one.
\end{lemma}

For $K_{x}$, we will use the following approximate (locally optimal) lower bound:
Given $H=\BC{n}{k}$, $G=G^{(K_{x})}$ and integer $\log_{n}{x}$ then a situation where all VMs in $K_{x}$ are arranged in the sub-cube $\BC{n}{\log_{n}{x}-1}$
constitutes a local minimum with cost is $(x+1)(n-1)/n$. The proof is by induction.


\subsection{Simulation Study}

To evaluate our \emph{Destination-Swap} algorithms in more detail and to compare their behavior in different settings and under
different communication patterns $\sigma$, we developed a simulation framework for the \BC{$n$}{$k$} datacenter topology.

If not stated otherwise, for our simulations, we will
consider a $\BC{3}{7}$ which consists of 6,561 nodes.
For each experiment, we simulate $|\sigma|=3$m requests.
We will repeat each experiment ten times and plot the average
values. (The variance of our experiments is typically very low.)

	\begin{figure}[h]
				\centering
				\includegraphics[width=.45\columnwidth]{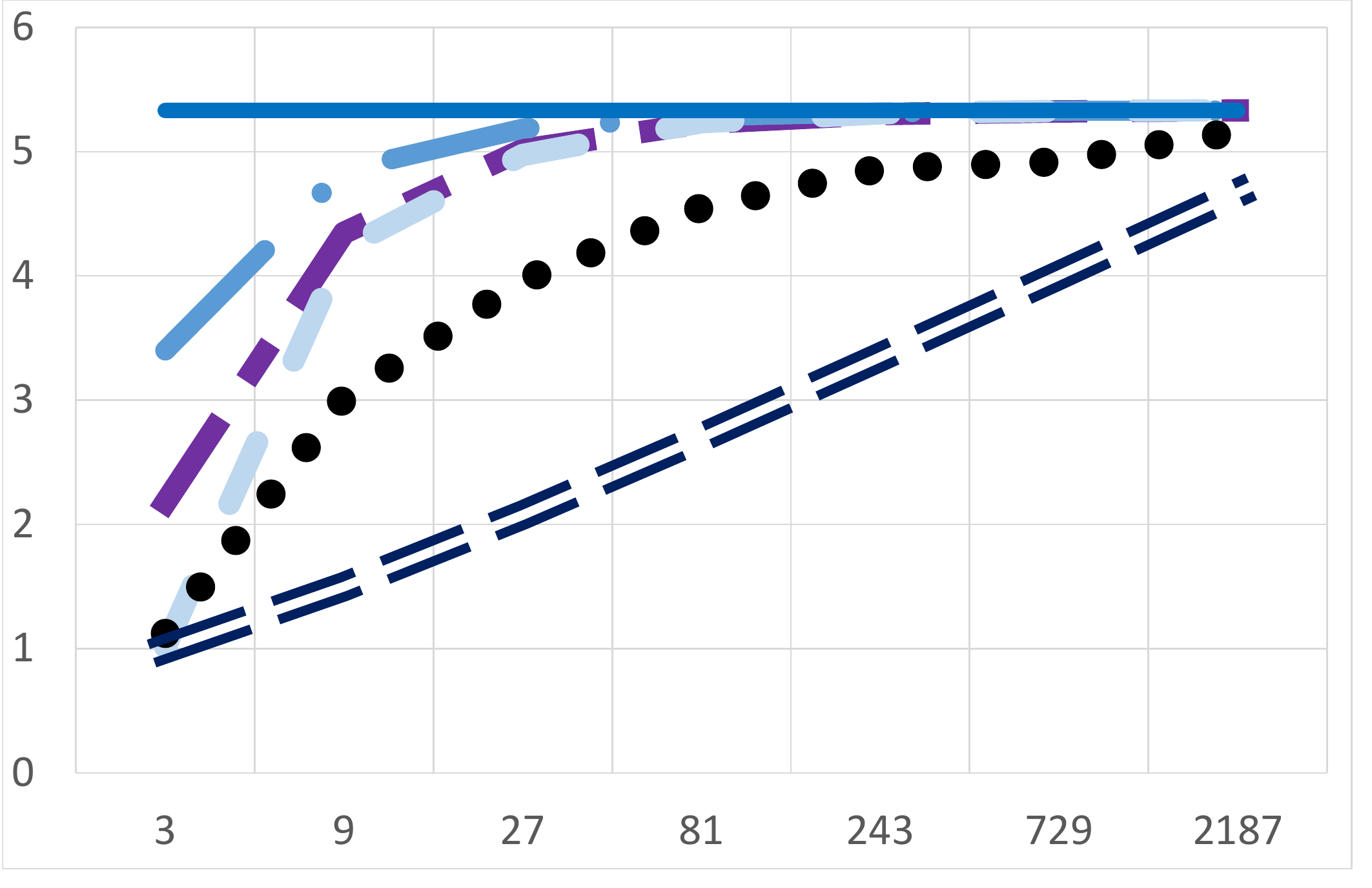}~				 \includegraphics[width=.45\columnwidth]{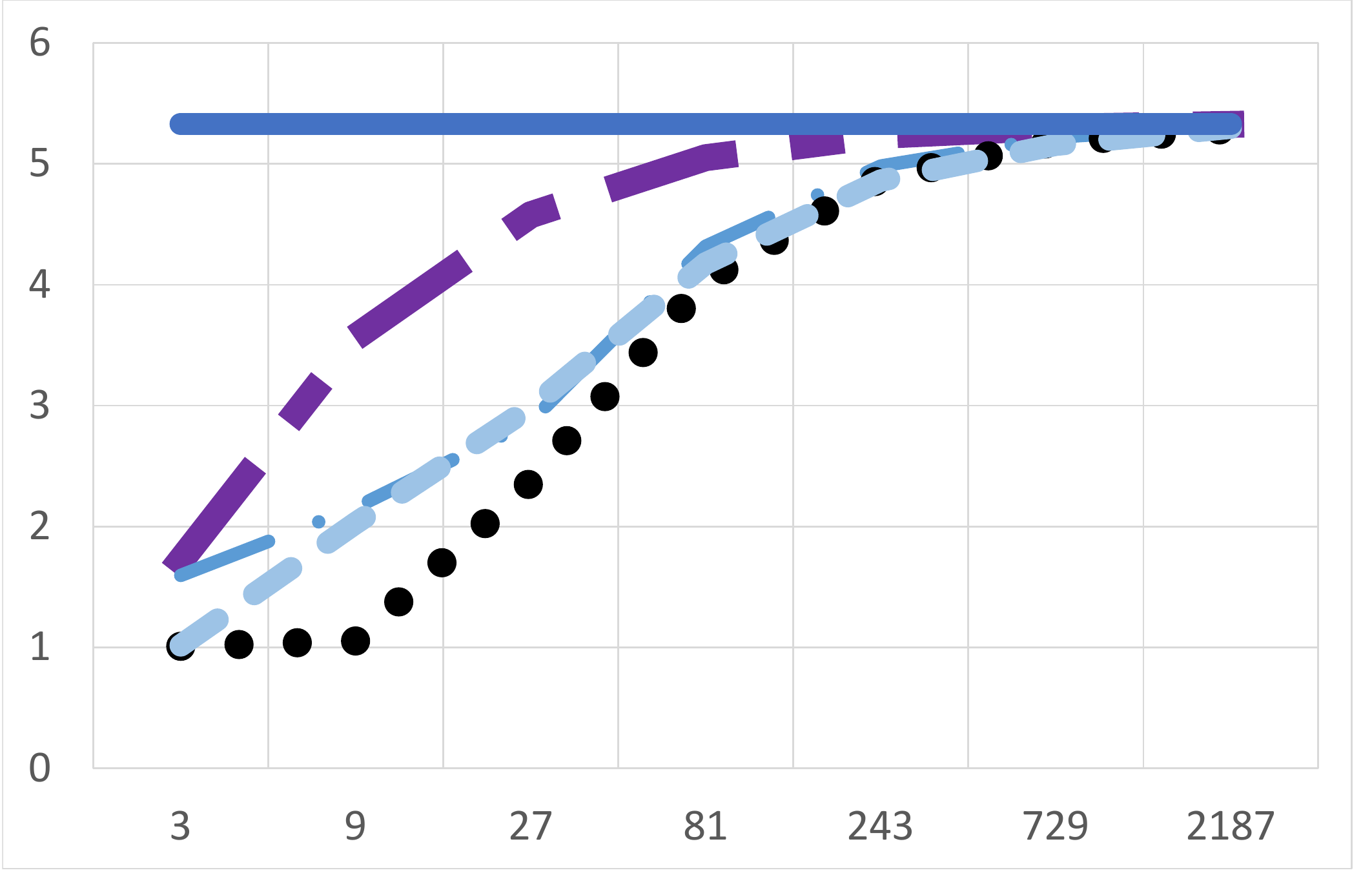}\\
				\includegraphics[width=.45\columnwidth]{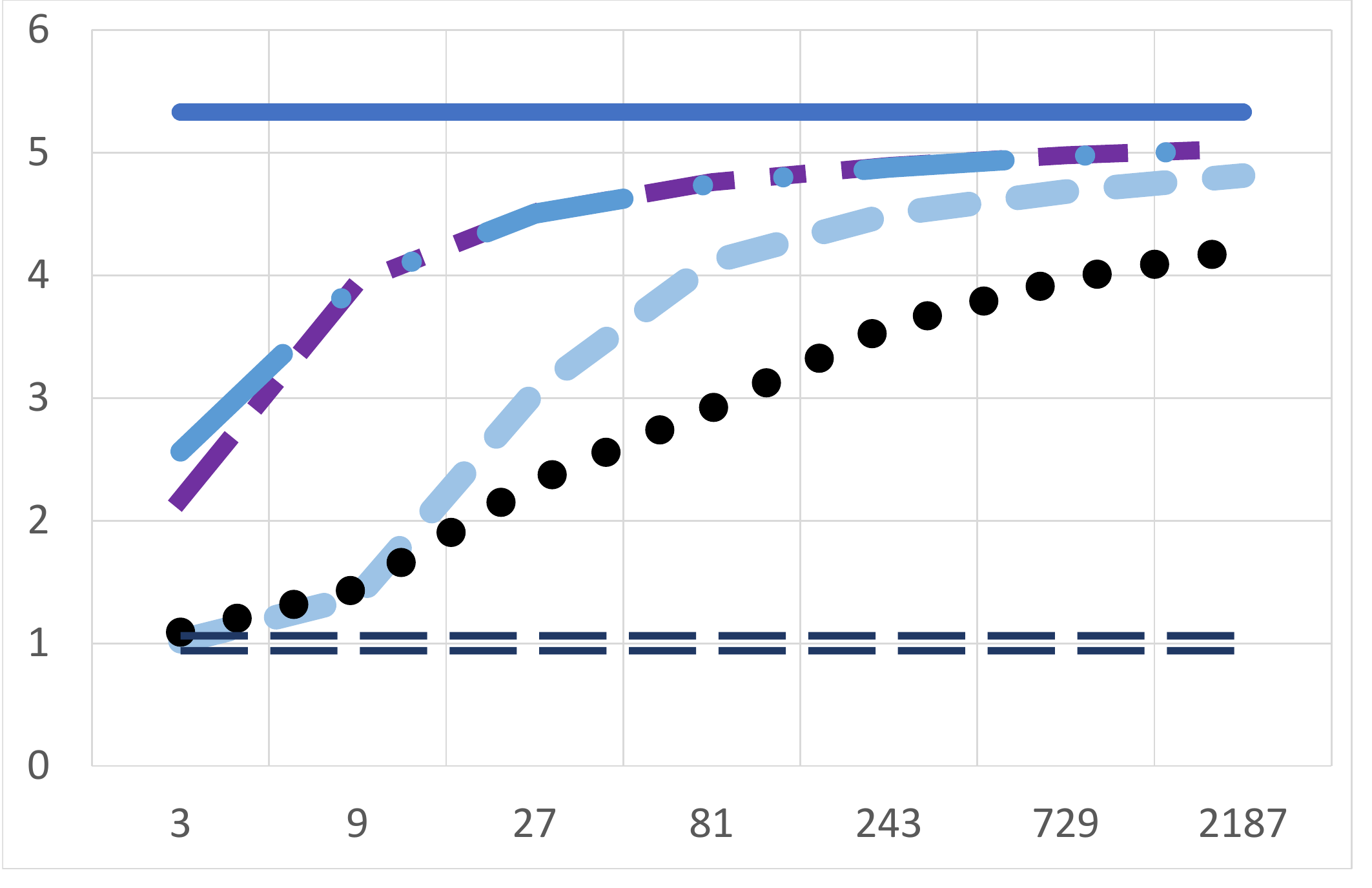}~      \includegraphics[width=.45\columnwidth, height=2.5cm]{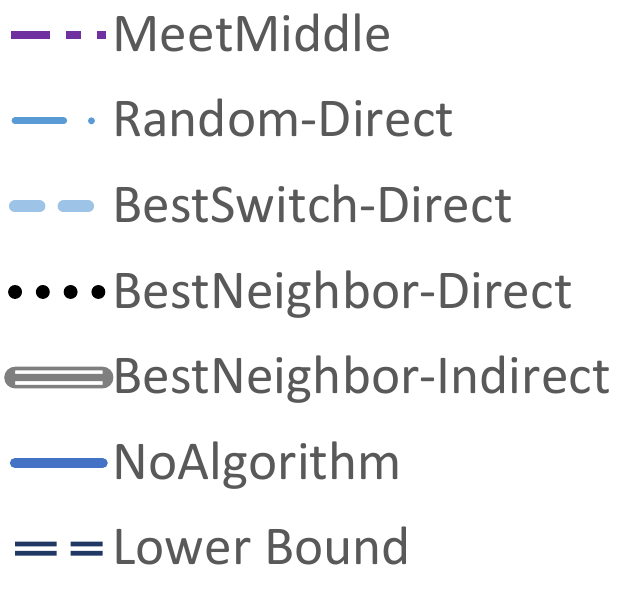}
				\caption{Amortized communication cost as a function of the (tenants) guest graph size after 3m requests for host graph \BC{3}{7} and for \emph{upper left:} all-to-all communication $G^{(K_{x})}$,
\emph{upper right:} one-to-all communication $G^{(S_{x})}$, \emph{bottom left:} $G^{(\BC{3}{\log_3(x)-1})}$ pattern.}
				\label{fig:First-3-8-Complete-WithOutWeights}
				\end{figure}
				
\subsubsection{Impact of Scale}

We first investigate how the amortized communication cost (over $\sigma$ requests) depends on the network size. We concentrate on \emph{direct} migration algorithms
for now.
Figure ~\ref{fig:First-3-8-Complete-WithOutWeights} (\emph{upper left}) shows that under unweighted all-to-all communication patterns $G^{(K_{x})}$ (as a function of $x$), all migration algorithms strictly improve the
amortized cost. The best performance is achieved by {\sc BestNeighbor}, followed by  {\sc BestSwitch}; {\sc Random} can be more than 50\% more expensive. This is not surprising, and shows that the additional information about the local amortized communication cost generally helps.
Moreover, since the neighbors of a specific VM can be located quite far from each other, the destination should be selected carefully for the swap operation.

Figure ~\ref{fig:First-3-8-Complete-WithOutWeights} (\emph{upper right}) studies the same setting but for a weighted one-to-all communication pattern $G^{(S_{x})}$. While the amortized costs are generally slightly lower in this case, the order of the algorithms is the same as before.
Finally, Figure~\ref{fig:First-3-8-Complete-WithOutWeights} (\emph{bottom left}) shows the results for weighted sub-cube communication patterns $G^{(\BC{n'}{k'})}$.

\subsubsection{Benefit and Limitation of Indirect Swaps}

We next consider the indirect swapping methods. Figure~\ref{fig:Second-3-8-Complete-WithOutWeights} (\emph{upper left}) (for the unweighted all-to-all communication $G^{(K_{x})}$)  and Figure~\ref{fig:Second-3-8-Complete-WithOutWeights} \emph{bottom} (for a weighted sub-cube communication $G^{(\BC{n'}{k'})}$)  provides a comparison of direct and indirect strategies of {\sc BestNeighbor}.
We see that direct swaps are better than indirect swaps: in case of highly connected guest graphs, the indirect strategy
migrates multiple VMs in a non-optimized manner, while the direct algorithm restricts itself to the best switch or neighbor.

	\begin{figure}[h]
				\centering \includegraphics[width=.45\columnwidth]{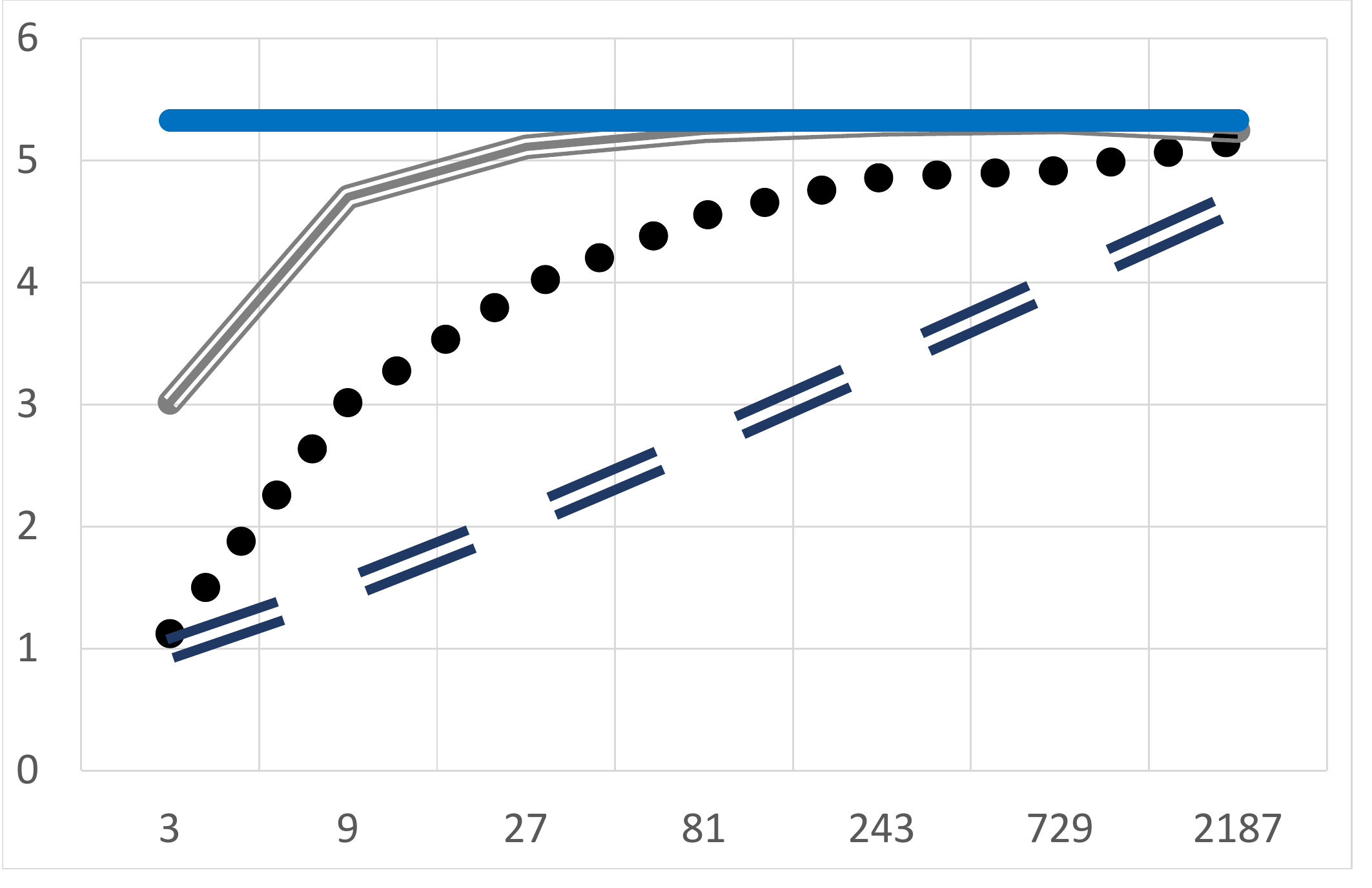}~\includegraphics[width=.45\columnwidth]{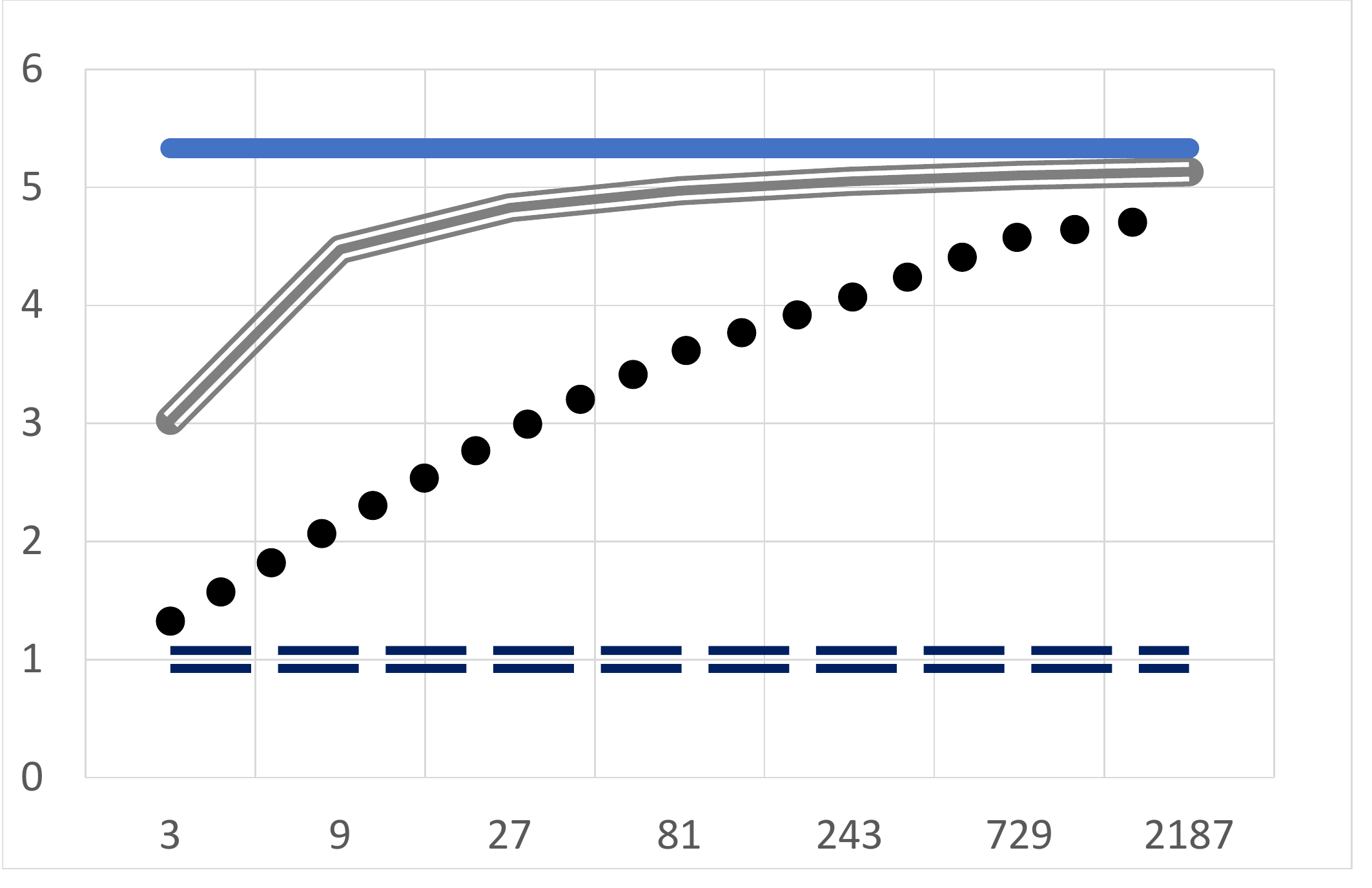} \\
				\includegraphics[width=.45\columnwidth]{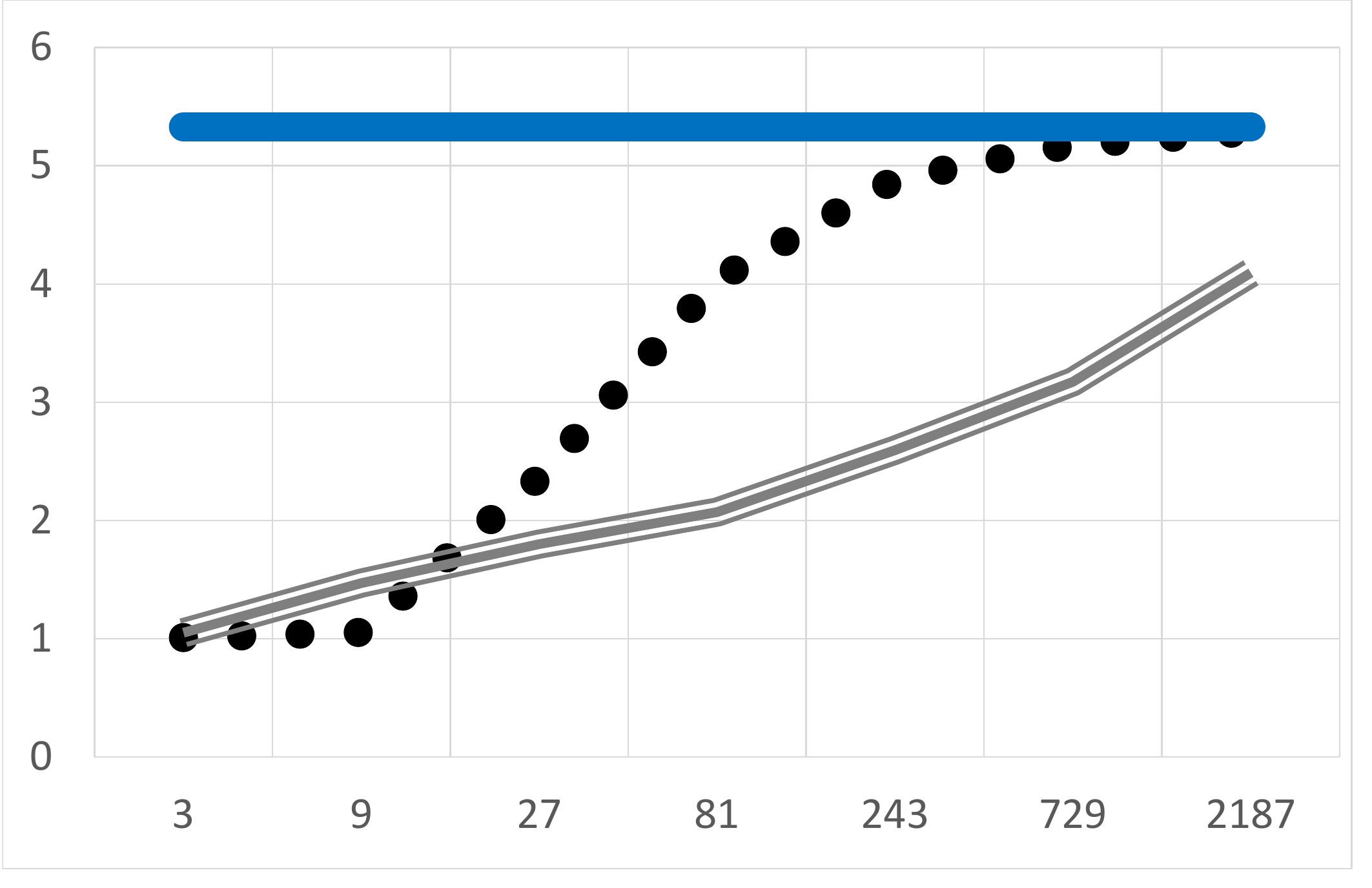}~
				\caption{Direct vs~indirect for {\sc BestNeighbor}: Amortized communication cost as a function of (tenant) guest graph size after 3m requests for host graph \BC{3}{7}  and \emph{upper left:} unweighted guest graph $G^{(K_{x})}$, \emph{upper right:}
 weighted $G^{(\BC{3}{\log_3(x)-1})}$ and \emph{bottom:} weighted one-to-all communication $G^{(S_{x})}$. The algorithm legend is the same as in Figure~\ref{fig:First-3-8-Complete-WithOutWeights}.}
				\label{fig:Second-3-8-Complete-WithOutWeights}
				\end{figure}

However, under a more star-like and weighted communication pattern, see Figure ~\ref{fig:Second-3-8-Complete-WithOutWeights} (\emph{bottom}) (for a weighted one-to-all communication $G^{(S_{x})}$)
the \emph{indirect} strategy is preferable. This can be explained by the fact that once {\sc BestNeighbor} managed to collocate most VMs of a given tenant,
the VMs communicating more frequently will stay closer to the center of the star guest graph; in contrast, in the direct approach, a frequently communicating
VM can be globally displaced again. This is also the reason why in Figure \ref{fig:Second-3-8-Complete-WithOutWeights} (\emph{bottom}), a larger guest graph (i.e., more involved VMs) increases the advantage of $\BestNeighborI$ over $\BestNeighborD$: from Lemma~\ref{lemma:No-Algo} we know that a host $m$ has  $(n-1)^1 {k+1 \choose 1} = 14$ neighbors at distance one, so when the guest graph is of size nine then also the direct approach can maintain a local communication.
However, when the guest graph becomes larger, more and more nodes need to be moved larger from the center and the advantages of the \emph{indirect} swaps are emphasized.




\begin{figure}[h]
				\centering
				 \includegraphics[width=.10\columnwidth, height=2.7cm]{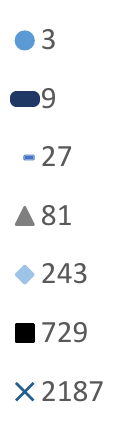}~\includegraphics[width=.45\columnwidth]{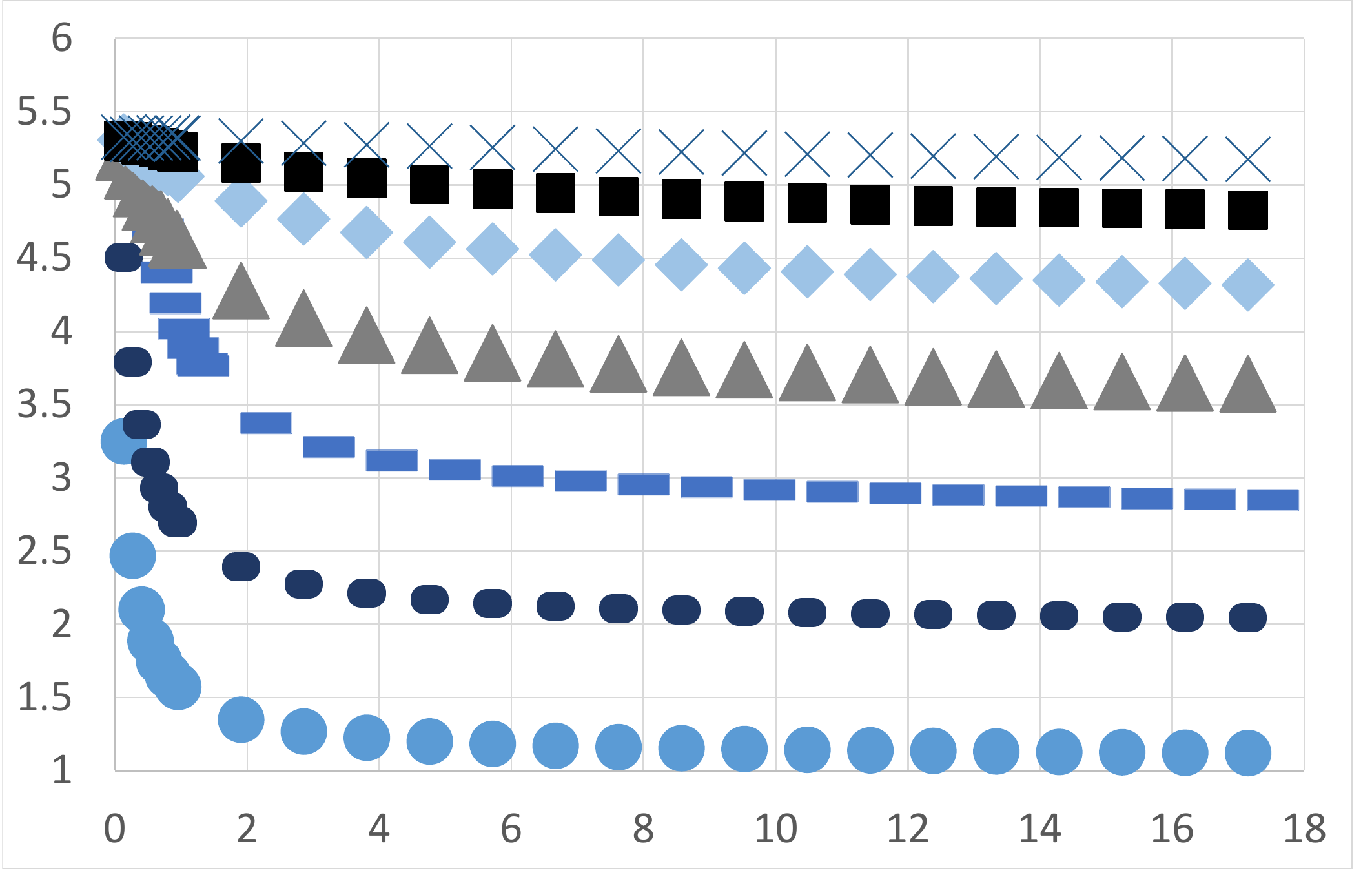}~\includegraphics[width=.45\columnwidth]{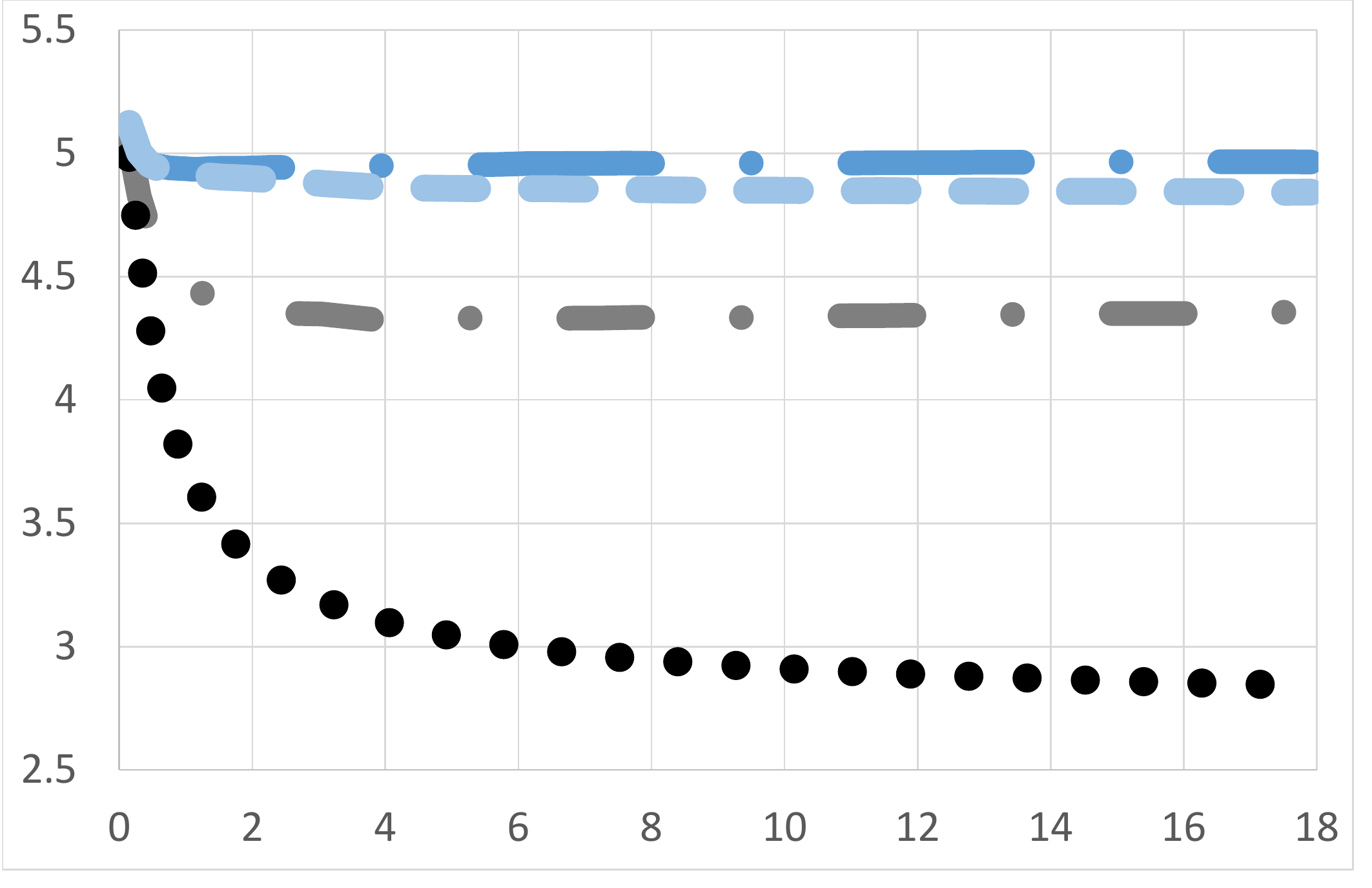}\\
				\caption{Amortized communication cost over time (number of per link requests) for host graph \BC{3}{7}. \emph{Left:} unweighted all-to-all guest graph $G^{(K_{x})}$ under \BestNeighborD (colors represent different sizes of tenant guest graph), \emph{right:} unweighted all-to-all guest graph $G^{(K_{27})}$ (colors represent different types of algorithms). The legend for the algorithms is the same as in Figure~\ref{fig:First-3-8-Complete-WithOutWeights}.}
				\label{fig:Time_first_complete}
				\end{figure}

\subsubsection{Reaction Time}
Of course, an eventual reduction of the amortized communication cost alone is not very interesting: over a long time,
a given communication pattern could also be learned and the embedding adapted accordingly. Rather, the main purpose
of our algorithms is to flexibly react to communication shifts and quickly find a new embedding.

In the following, we provide evidence that our algorithms indeed readjust the embedding quickly. Figure ~\ref{fig:Time_first_complete} \emph{left} show that the convergence is quick for the different kinds of our algorithms, and Figure ~\ref{fig:Time_first_complete} \emph{right} show that it's also quick for different sizes of tenants graphs (both under unweighted all-to-all communication); the time axis is normalized and represents the number of requests per guest graph edge.
We find that the convergence time of all our algorithm is very low: for a 6,561-nodes host graph, a good embedding is found after around 10 requests per edge. Other simulation results (not included in these figures) show that for weighted one-to-all communication the convergence if about after 3 requests per edge.

\subsubsection{Intra-Tenant Migration}

\begin{wrapfigure}{l}{0.5\columnwidth}
				\centering
				\includegraphics[width=.45\columnwidth]{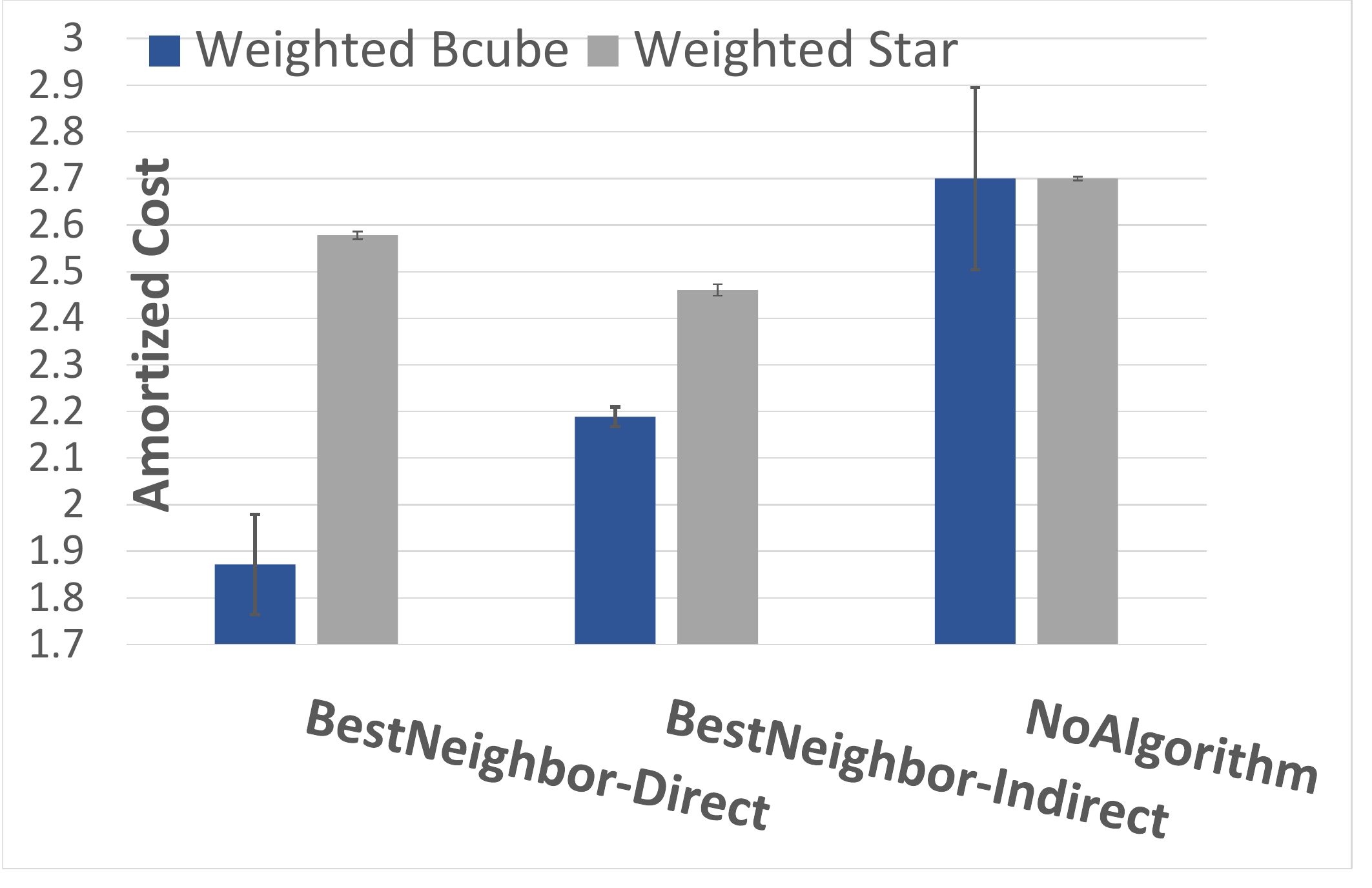}
				\caption{Amortized communication cost on host graph \BC{3}{3}, and a BCube as well as a
star communication pattern (weighted).}
				\label{fig:Third_4-3_PublicDC}
\end{wrapfigure}
Our algorithms can also be used to optimize the VM placement inside a tenant: we restrict the migration destinations and neighbors
to the VMs of the given tenant. Naturally, this reduces the number of migration options significantly, and we expect lower gains from migration.
In order to investigate the benefits and limitations of intra-tenant migrations, we consider a setting where the substrate network
hosts a single, connected tenant. This allows us to abstract from cost artefacts due to non-local VMs placements: the single tenant
is mapped locally, but in a random and hence suboptimal manner.
Figure~\ref{fig:Third_4-3_PublicDC} shows our results for host graph $\BC{3}{3}$ and
a BCube and a star communication pattern. The figure indicates that
even inside a single tenant, migration can help reducing the communication
cost by around 50\% under a BCube traffic matrix; under a one-to-all
communication pattern, the benefits are lower.


\section{Related Work}\label{sec:relwork}

Today's networks become more and more dynamic in the sense that they are able to self-adjust to the network state, user demand,
or even energy cost, and the benefits of process migration have been exploited long before
the emergence of the cloud computing paradigm, e.g., for load-balancing
applications~\cite{migration-early}. Generally, applications range from
self-optimizing peer-to-peer topologies over green computing (e.g.,
due to reduced energy consumption)~\cite{elastictree} to
adaptive virtual machine migrations in
datacenters~\cite{Shang2010Energy-aware}, microprocessor memory
architectures~\cite{lis2011brief}, grids~\cite{batista2007self} or elastic virtual and wide-area cloud networks~\cite{ucc12mip}.
Other self-adjusting routing scheme were considered, e.g., in
scale-free networks to overcome congestion \cite{PhysRevE.80.026114}.

VM migration in cloud computing has been proposed to improve resource utilization as well as to balance load
and alleviate hotspots~\cite{seffi-vm-migration}, or even save energy~\cite{elastictree}. For example, the VMware Distributed
Resource Scheduler uses live migration to balance load in
response to CPU and memory contention.

Seamless VM migration can be implemented in different ways.
One may \emph{pre-copy}~\cite{ref3} the
VM memory to the destination host before releasing
the VM at the source,  or defer the memory transfer until the processor state has been sent
to the destination (the so-called \emph{post-copy}~\cite{ref9} approach).
Wood et al.~\cite{sandpiper} investigate automated black-box, gray-box, and hybrid strategies for VM
migration in datacenters.
Researchers also conduct measurements and derive models for VM migration costs, e.g., under Web 2.0 workloads and quality-of-service sensitive applications.~\cite{mess}

VM migration has also been proposed for wide-area networks, where a lazy copyon
reference can be used for moving VM disk state to reduce migration costs over weak links.~\cite{wood-wan}
In the wide-area, moving entire \emph{services} closer to the (mobile) users can reduce access latency~\cite{migton}, and there
 also exists work on the migration of entire \emph{virtual networks}~\cite{ucc12mip} which are latency-critical (``move-with-the-sun'')
 or latency-uncritical (``move-with-the-moon''). Also in the context of network virtualization,
Hao et al.~\cite{cloud-netvirt} have shown that under certain circumstances,
the migration
of a Samba front-end server closer to the clients can be beneficial
even for bulk-data applications.

In the theory community, many migration
problem variants have been studied in the context of
online page migration, and more generally, online metrical task systems~\cite{borodin}.
There exist several interesting results by Naor et al.~on how to embed
and migrate services to reduce load, e.g.,~\cite{podc11,seffi-mini}.
Indeed, given a fixed communication pattern between pairs of VMs, our work is related to classic
graph embedding problems such as the \emph{minimum linear arrangement} (e.g.,~\cite{mla}) of a graph.
Recently, such embedding problems have also been studied from the perspective of self-adjusting
networks, in the context of distributed splay datastructures and peer-to-peer networks~\cite{ipdps13}.
Indeed, our algorithms are inspired by the classic splaying techniques introduced
in the seminal work by Sleator and Tarjan~\cite{Sleator:1985:SBS:3828.3835} on self-adjusting search trees,
in the sense that we also aggressively migrate VMs closer together.

However, we are not aware on any literature on online and adaptive VM migration algorithms that flexibly
adapt to a dynamic demand.


\section{Conclusion}\label{sec:conclusion}

We understand our paper as a first step towards a better understanding of simple and adaptive VM
migration strategies. The presented Destination-Swap algorithms are pair-based only and do not try to infer or accumulate
much information about the environment. This slim solution increases the flexibility, and hence reduces the reaction time of the
algorithms.

We have shown that this approach can indeed improve the embedding of VMs, especially in inter-tenant situations where VMs can be globally
optimized throughout the entire (private) datacenter. If a tenant must improve the embeddings itself and is restricted to a given set
of physical machines, the benefits are naturally lower but still visible. Moreover, our results suggest that there exists an interesting
tradeoff between direct and indirect approaches, and we find that indirect algorithms are better for sparser and more skewed
communication patterns, as locality is preserved.

Our work opens several interesting directions for future research. Obviously, a deeper understanding of the tradeoff between
the embedding quality and the number of migrations needs to be developed. Another interesting question regards the amount of
additional information (beyond the local amortized cost) needed to improve the mapping further and faster.

{\footnotesize
\renewcommand{\baselinestretch}{.83}
\bibliographystyle{abbrv}
\bibliography{migration}

\begin{thebibliography}{10}

\bibitem{ipdps13}
C.~Avin, B.~Haeupler, Z.~Lotker, C.~Scheideler, and S.~Schmid.
\newblock Locally self-adjusting tree networks.
\newblock In {\em Proc. IEEE IPDPS}, 2013.

\bibitem{mesos}
{B.~Hindman et al.}
\newblock Mesos: a platform for fine-grained resource sharing in the data
  center.
\newblock In {\em Proc. USENIX NSDI}, 2011.

\bibitem{oktopus}
H.~Ballani, P.~Costa, T.~Karagiannis, and A.~Rowstron.
\newblock Towards predictable datacenter networks.
\newblock In {\em Proc. ACM SIGCOMM}, pages 242--253, 2011.

\bibitem{podc11}
N.~Bansal, K.-W. Lee, V.~Nagarajan, and M.~Zafer.
\newblock Minimum congestion mapping in a cloud.
\newblock In {\em Proc. ACM PODC}, 2011.

\bibitem{batista2007self}
D.~Batista, N.~da~Fonseca, F.~Granelli, and D.~Kliazovich.
\newblock Self-adjusting grid networks.
\newblock In {\em Proc. ICC}, 2007.

\bibitem{generalizedHyper}
L.~N. Bhuyan and D.~P. Agrawal.
\newblock Generalized hypercube and hyperbus structures for a computer network.
\newblock {\em IEEE Trans. Comput.}, 33(4):323--333, Apr. 1984.

\bibitem{borodin}
A.~Borodin and R.~El-Yaniv.
\newblock {\em Online computation and competitive analysis}.
\newblock Cambridge University Press, New York, NY, USA, 1998.

\bibitem{ref3}
{C. Clark et al.}
\newblock Live migration of virtual machines.
\newblock In {\em Proc. NSDI}, 2005.

\bibitem{seffi-mini}
R.~Cohen, L.~Lewin-Eytan, S.~Naor, and D.~Raz.
\newblock Almost optimal virtual machine placement for traffic intense data
  center.
\newblock In {\em Proc. IEEE Infocom Mini-Conference}, 2013.

\bibitem{mla}
N.~R. Devanur, S.~A. Khot, R.~Saket, and N.~K. Vishnoi.
\newblock Integrality gaps for sparsest cut and minimum linear arrangement
  problems.
\newblock In {\em Proc. ACM STOC}, 2006.

\bibitem{nfv}
ETSI.
\newblock Network functions virtualisation.
\newblock {\em Introductory White Paper}, 2013.

\bibitem{bcube2009}
C.~Guo, G.~Lu, D.~Li, H.~Wu, X.~Zhang, Y.~Shi, C.~Tian, Y.~Zhang, and S.~Lu.
\newblock Bcube: a high performance, server-centric network architecture for
  modular data centers.
\newblock {\em SIGCOMM CCR}, 39(4):63--74, Aug. 2009.

\bibitem{cloud-netvirt}
F.~Hao, T.~V. Lakshman, S.~Mukherjee, and H.~Song.
\newblock Enhancing dynamic cloud-based services using network virtualization.
\newblock {\em SIGCOMM CCR}, 40(1):67--74, 2010.

\bibitem{migration-early}
M.~Harchol-Balter and A.~B. Downey.
\newblock Exploiting process lifetime distributions for dynamic load balancing.
\newblock {\em ACM Trans. Comput. Syst.}, 15(3):253--285, 1997.

\bibitem{elastictree}
B.~Heller, S.~Seetharaman, P.~Mahadevan, Y.~Yiakoumis, P.~Sharma, S.~Banerjee,
  and N.~McKeown.
\newblock Elastictree: saving energy in data center networks.
\newblock In {\em Proc. USENIX NSDI}, 2010.

\bibitem{ref9}
M.~R. Hines and K.~Gopalan.
\newblock Post-copy based live virtual machine migration using adaptive
  pre-paging and dynamic self-ballooning.
\newblock In {\em Proc. VEE}, 2009.

\bibitem{seffi-vm-migration}
N.~Jain, I.~Menache, J.~S. Naor, and F.~B. Shepherd.
\newblock Topology-aware vm migration in bandwidth oversubscribed datacenter
  networks.
\newblock In {\em Proc. 39th ICALP}, 2012.

\bibitem{lis2011brief}
M.~Lis, K.~Shim, M.~Cho, C.~Fletcher, M.~Kinsy, I.~Lebedev, O.~Khan, and
  S.~Devadas.
\newblock Brief announcement: distributed shared memory based on computation
  migration.
\newblock In {\em SPAA}, pages 253--256. ACM, 2011.

\bibitem{migton}
{M. Bienkowski et al.}
\newblock The wide-area virtual service migration problem: A competitive
  analysis approach.
\newblock {\em IEEE/ACM Transactions on Networking (ToN)}, to appear.

\bibitem{talk-about}
J.~C. Mogul and L.~Popa.
\newblock What we talk about when we talk about cloud network performance.
\newblock {\em SIGCOMM CCR}, 42(5):44--48, 2012.

\bibitem{qclouds}
R.~Nathuji, A.~Kansal, and A.~Ghaffarkhah.
\newblock Q-clouds: Managing performance interference effects for qos-aware
  clouds.
\newblock In {\em Proc. EuroSys}, 2010.

\bibitem{intra-tenant}
Z.~Ou, H.~Zhuang, J.~K. Nurminen, A.~Yl\"{a}-J\"{a}\"{a}ski, and P.~Hui.
\newblock Exploiting hardware heterogeneity within the same instance type of
  amazon ec2.
\newblock In {\em Proc. HotCloud}, 2012.

\bibitem{getoff}
T.~Ristenpart, E.~Tromer, H.~Shacham, and S.~Savage.
\newblock Hey, you, get off of my cloud: exploring information leakage in
  third-party compute clouds.
\newblock In {\em Proc. 16th ACM CCS}, pages 199--212, 2009.

\bibitem{ucc12mip}
G.~Schaffrath, S.~Schmid, and A.~Feldmann.
\newblock Optimizing long-lived cloudnets with migrations.
\newblock In {\em Proc. IEEE/ACM UCC}, 2012.

\bibitem{Shang2010Energy-aware}
Y.~Shang, D.~Li, and M.~Xu.
\newblock Energy-aware routing in data center network.
\newblock In {\em Proc. Workshop Green Networking}, pages 1--8, New York, NY,
  USA, 2010.

\bibitem{Sleator:1985:SBS:3828.3835}
D.~D. Sleator and R.~E. Tarjan.
\newblock Self-adjusting binary search trees.
\newblock {\em J. ACM}, 32(3):652--686, July 1985.

\bibitem{PhysRevE.80.026114}
M.~Tang, Z.~Liu, X.~Liang, and P.~M. Hui.
\newblock Self-adjusting routing schemes for time-varying traffic in scale-free
  networks.
\newblock {\em Phys. Rev. E}, 80(2):026114, Aug 2009.

\bibitem{mess}
W.~Voorsluys, J.~Broberg, S.~Venugopal, and R.~Buyya.
\newblock Cost of virtual machine live migration in clouds: A performance
  evaluation.
\newblock In {\em Proc. International Conference on Cloud Computing
  (CloudCom)}, pages 254--265, 2009.

\bibitem{amazon-per}
G.~Wang and E.~Ng.
\newblock The impact of virtualization on network performance of amazon ec2
  data center.
\newblock In {\em Proc. INFOCOM}, 2010.

\bibitem{EC2over}
A.~Williamson.
\newblock Has amazon {EC2} become over subscribed?
\newblock
  \url{http://alan.blog-city.com/has_amazon_ec2_become_over_subscribed.htm},
  2013.

\bibitem{wood-wan}
T.~Wood, K.~K. Ramakrishnan, P.~Shenoy, and J.~van~der Merwe.
\newblock Cloudnet: dynamic pooling of cloud resources by live wan migration of
  virtual machines.
\newblock {\em SIGPLAN Not.}, 46(7):121--132, 2011.

\bibitem{sandpiper}
T.~Wood, P.~Shenoy, A.~Venkataramani, and M.~Yousif.
\newblock Black-box and gray-box strategies for virtual machine migration.
\newblock In {\em Proc. USENIX NSDI}, 2007.

\end{thebibliography}
}
\end{document}